\documentclass[11pt]{article}

\usepackage[pagebackref,colorlinks]{hyperref}
\usepackage{amsmath} 
\usepackage{amsthm} 
\usepackage{thmtools}
\usepackage{amssymb}	
\usepackage{graphicx} 
\usepackage{multicol} 
\usepackage{multirow}
\usepackage{color}
\usepackage{bm}
\usepackage{enumitem}
\usepackage[dvips,letterpaper,margin=1in,bottom=1in]{geometry}
\usepackage[capitalize,noabbrev]{cleveref}
\usepackage{booktabs}

\usepackage[utf8]{inputenc}
\usepackage[english]{babel}

\usepackage{mathtools}

\newtheorem{theorem}{Theorem}[section]

\newtheorem{lemma}[theorem]{Lemma}

\newtheorem{remark}{Remark}[section]

\newcommand{\braket}[2]{\left< #1 \vphantom{#2} \middle| #2 \vphantom{#1} \right>} 
\newcommand{\ketbra}[2]{\ensuremath{\ket{#1}\!\bra{#2}}}

\DeclarePairedDelimiter\rbra{\lparen}{\rparen}
\DeclarePairedDelimiter\sbra{\lbrack}{\rbrack}
\DeclarePairedDelimiter\cbra{\{}{\}}
\DeclarePairedDelimiter\abs{\lvert}{\rvert}
\DeclarePairedDelimiter\Abs{\lVert}{\rVert}
\DeclarePairedDelimiter\ceil{\lceil}{\rceil}

\DeclarePairedDelimiter\ket{\lvert}{\rangle}
\DeclarePairedDelimiter\bra{\langle}{\rvert}

\newcommand{\tr} {\operatorname{tr}}
\newcommand{\E} {\mathbf{E}}

\newcommand{\Real} {\operatorname{Re}}

\DeclarePairedDelimiter\parens{\lparen}{\rparen}
\DeclarePairedDelimiter\norm{\lVert}{\rVert}
\DeclarePairedDelimiter\braces{\lbrace}{\rbrace}

\newcommand{\calE}{\mathcal{E}}

\newcommand{\calI}{\mathcal{I}}

\usepackage{algorithm}
\usepackage{algpseudocode}

\usepackage{tabularx}
\usepackage{booktabs}
\usepackage{threeparttable}
\usepackage{adjustbox} 

\usepackage{tikz}
\usetikzlibrary{quantikz2}

\title{Strict Hierarchy for Quantum Channel Certification to Unitary}
\author{Kean Chen\thanks{\url{keanchen.gan@gmail.com}}\and
Qisheng Wang\thanks{\url{QishengWang1994@gmail.com}} \and 
Zhicheng Zhang\thanks{\url{iszczhang@gmail.com}}}
\date{}

\begin{document}

\maketitle

\begin{abstract}
We consider the problem of quantum channel certification to unitary, where one is given access to an unknown $d$-dimensional channel $\mathcal{E}$, and wants to test whether $\mathcal{E}$ is equal to a target unitary channel or is $\varepsilon$-far from it in the diamond norm. 
We present optimal quantum algorithms for this problem, settling the query complexities in three access models with increasing power.
Specifically, we show that:
\begin{enumerate}[label=(\roman*)]
    \item $\Theta(d/\varepsilon^2)$ queries suffice for incoherent access model, matching the lower bound due to \hyperlink{cite.FFGO23}{Fawzi, Flammarion, Garivier, and Oufkir (COLT 2023)}.
    \item $\Theta(d/\varepsilon)$ queries suffice for coherent access model, matching the lower bound due to \hyperlink{cite.RS08}{Regev and Schiff (ICALP 2008)}. 
    \item $\Theta(\sqrt{d}/\varepsilon)$ queries suffice for source-code access model, matching the lower bound due to \hyperlink{cite.JO26}{Jeon and Oh (\textit{npj Quantum Inf.}\ 2026)}. 
\end{enumerate}
This demonstrates a strict hierarchy of complexities for quantum channel certification to unitary across various access models.
\end{abstract}

\section{Introduction}
A fundamental task in quantum computing and quantum information is
to certify whether a quantum device implements a desired operation. 
The certification of quantum objects plays an important role in quantum property testing \cite{MdW16}, including the certification of quantum states \cite{OW21,BOW19,chen2022toward,CLHL22}, quantum Hamiltonians \cite{GJW+26,LS25}, and quantum channels \cite{FFGO23,RAS+24,JO26}.

In this paper, we consider quantum channel certification.
In principle, one can learn a full classical description of an unknown quantum channel via channel tomography~\cite{haah2023query,Oufkir_2023,oufkir2023adaptivity,surawy2022projected,RAS+24,yoshida2025quantum,mele2025optimal,chen2025quantum}. However, full tomography requires far more resources than the certification task and is generally inefficient.
A more natural and widely studied goal is to decide only whether the device meets a given specification or is far from it.

\paragraph{Quantum channel certification to unitary.}
More formally, we study the problem of \emph{quantum channel certification to unitary}~\cite{FFGO23}.\footnote{The work \cite{FFGO23} considers quantum channel certification to identity, which is essentially the same task as quantum channel certification to unitary considered in this paper. In addition, they also consider the certification in the trace norm, a different closeness measure of quantum channels.}
Given a known unitary channel $\mathcal{U}(\rho)=U\rho U^\dagger$ and query access to an unknown quantum channel $\mathcal{E}$,
the goal is to decide between
\[
\textup{Case 1}\colon \mathcal{E}=\mathcal{U}
\qquad\text{and}\qquad
\textup{Case 2}\colon \|\mathcal{E}-\mathcal{U}\|_\diamond \ge \varepsilon,
\]
with success probability $\geq 2/3$, while minimizing the number of queries to $\mathcal{E}$.
Here, we consider the certification in the diamond norm $\norm*{\cdot}_{\diamond}$, which captures the strongest black-box notion of
closeness to a target operation.

\paragraph{Access models and a strict complexity hierarchy.}
A key lesson from recent work~\cite{BOW19,huang2020predicting,bubeck2020entanglement,buadescu2021improved,huang2021information,aharonov2022quantum,chen2022exponential,chen2022toward,chen2023unitarity,FFGO23,chen2024local,TW25,tang2025conjugate} in quantum learning and testing is that is that the access model for the unknown object matters: quantum memory, (in)coherent controls, and variants of oracle queries can yield provably different complexity results. In this paper, we consider three standard access models with increasing power: incoherent access, coherent access, and source-code access.

\begin{enumerate}
  \item \textbf{Incoherent access} (cf.\ \cite{chen2022exponential}) corresponds to quantum algorithms without quantum memory. More precisely,
    after each query to $\mathcal{E}$, the algorithm must perform a measurement; only classical information may be retained and used in subsequent queries.
    Formally, an incoherent algorithm using $T$ queries to $\mathcal{E}$ has the following behavior: When using the $i$-th query to $\mathcal{E}$ for $1 \leq i \leq T$, the algorithm prepares a quantum state $\rho_i$ according to the classical information $C_1, C_2, \dots, C_{i-1}$ obtained previously, and applies $\mathcal{E}$ to $\rho_i$, followed by a measurement with outcome $C_i$. 
    
  \item \textbf{Coherent access}, contrast to incoherent access, corresponds to quantum algorithms with quantum memory. The algorithm can perform an arbitrary joint quantum
        computation across multiple queries to $\mathcal{E}$.
        
  \item \textbf{Source-code access} (cf.\ \cite{KO23}) corresponds to coherent access and, additionally, access to the ``source code'' for implementing the target channel $\calE$.
  Here, the ``source code'' describes a quantum circuit $W$ that implements the unknown oracle (in our case, the channel $\calE$).
    More precisely, the algorithm can query the unitary operator $W$ (and its inverse $W^\dagger$) such that $\tr_{\mathsf{B}}\parens*{W(\rho_{\mathsf{A}}\otimes \ket{0}\!\bra{0}_{\mathsf{B}}) W^\dagger}=\calE(\rho_{\mathsf{A}})$ for any input state $\rho_{\mathsf{A}}$,
  where $\mathsf{A}$ is the main system and $\mathsf{B}$ is the environment.
\end{enumerate}

The main result of this paper is a strict complexity hierarchy of complexities for quantum channel certification to unitary across the above three models, showing a \textit{strict increase in power} from incoherent to coherent to source-code access.

\begin{table}[t]
\centering
\begin{tabular}{lccc}
\toprule
\textbf{Access Model} & \textbf{Query Complexity} & \textbf{Upper Bounds} & \textbf{Lower Bounds} \\
\midrule
\addlinespace
Incoherent Access & $\Theta(d/\varepsilon^{2})$ & \Cref{thm-222140} &\cite{FFGO23} \\
\addlinespace
Coherent Access & $\Theta(d/\varepsilon)$ & \Cref{thm-230318} &\cite{RS08} \\
\addlinespace
Source-Code Access & $\Theta(\sqrt{d}/\varepsilon)$ & \Cref{thm-1526} &\cite{JO26} \\
\bottomrule
\end{tabular}
\caption{Query complexity of certifying $\mathcal{E}=\mathcal{U}$ versus $\|\mathcal{E}-\mathcal{U}\|_\diamond\ge\varepsilon$
under three access models.}
\label{tab:main}
\end{table}

\subsection{Main results}
In this paper, we provide efficient algorithms for quantum channel certification to unitary in each of the three access models mentioned above.
Combining with the existing lower bounds in \cite{FFGO23,RS08,JO26}, we completely characterize the query complexity of this problem.

\begin{theorem}[Informal, \cref{thm-222140,thm-230318,thm-1526} restated] \label{thm-0220}
Let $\mathcal{E}$ be an unknown $d$-dimensional quantum channel and $\mathcal{U}$ be a known
$d$-dimensional unitary channel. The optimal number of queries needed to certify
$\mathcal{E}=\mathcal{U}$ versus $\|\mathcal{E}-\mathcal{U}\|_\diamond\ge \varepsilon$ is:
\begin{enumerate}
  \item $\Theta(d/\varepsilon^2)$ in the incoherent access model;
  \item $\Theta(d/\varepsilon)$ in the coherent access model;
  \item $\Theta(\sqrt{d}/\varepsilon)$ in the source-code access model.
\end{enumerate}
\end{theorem}

All three bounds in \cref{thm-0220} are tight.
For incoherent access, the lower bound is due to Fawzi, Flammarion, Garivier, and Oufkir~\cite{FFGO23}, and our upper bound improves their $O\rbra{d/\varepsilon^4}$.
For coherent access, the lower bound follows from the proof and the hard instance in work of Regev and
Schiff~\cite{RS08}.
For source-code access, our upper bound matches the lower bound recently established by Jeon and Oh~\cite{JO26}.
Together, these results exhibit a \emph{strict} complexity hierarchy for the same certification task under
increasingly powerful access to the unknown channel.

\begin{remark}
    In \cite[Theorem 3]{JO26}, they considered a similar and related problem of \textit{unitary channel certification to identity}, where the task is to determine whether $\mathcal{U} = \mathcal{I}$ or $\Abs{\mathcal{U} - \mathcal{I}}_\diamond \geq \varepsilon$ for an unknown unitary channel $\mathcal{U}$, given query access to $\mathcal{U}$ and $\mathcal{U}^{-1}$. They provided a quantum algorithm for this problem with query complexity $\Theta\rbra{\sqrt{d}/\varepsilon}$. 
    In comparison, our result in the source-code access model implies their upper bound, as our problem ``quantum channel certification to unitary'' is more general: the unknown channel $\mathcal{E}$ is not guaranteed to be unitary. 
\end{remark}

\subsection{Techniques}

We begin with our optimal approach for incoherent access, based on which we further obtain our approaches for coherent and source-code accesses. 

\paragraph{Incoherent access.}

We leverage the Choi--Jamio{\l}kowski isomorphism: applying
    $(\mathcal{U}^{-1}\circ \mathcal{E})\otimes \mathcal{I}$ to the maximally entangled state $\ket{\Phi}$ and measuring the
    projector onto $\ket{\Phi}$ yields a Bernoulli random variable $X$ with expectation $\E\sbra{X} = \mathrm{F}_{\textup{ent}}\rbra{\mathcal{E}, \mathcal{U}}$, where $\mathrm{F}_{\textup{ent}}\rbra{\mathcal{E}, \mathcal{U}}$ is the \textit{entanglement fidelity} \cite{Nie02} (or \textit{channel fidelity}~\cite{raginsky2001fidelity}) between $\mathcal{E}$ and $\mathcal{U}$.\footnote{For general quantum channels, the entanglement fidelity is defined as $\mathrm{F}_{\textup{ent}}\rbra{\mathcal{E}, \mathcal{F}}\coloneqq \mathrm{F}((\mathcal{E}\otimes \mathcal{I})(\ketbra{\Phi}{\Phi}), (\mathcal{F}\otimes \mathcal{I})(\ketbra{\Phi}{\Phi}))$, in which $\mathrm{F}$ denotes the (squared) fidelity between quantum states.} between $\mathcal{E}$ and $\mathcal{U}$.
    To this end, we present a quantitative link between the entanglement fidelity and the diamond norm (in \cref{lemma-1181906}):
\[
\mathrm{F}_{\textup{ent}}\rbra{\mathcal{E}, \mathcal{U}} \leq \sqrt{\mathrm{F}_{\textup{ent}}\rbra{\mathcal{E}, \mathcal{U}}} \leq 1 - \frac{1}{8d} \Abs{\mathcal{E} - \mathcal{U}}_\diamond^2.
\]
    Further analysis shows that (i) $\E\sbra{X} = 1$ when $\mathcal{E} = \mathcal{U}$ while (ii) $\E\sbra{X} \leq 1 - \Theta\rbra{{\varepsilon^2}/{d}}$ when $\Abs{\mathcal{E}-\mathcal{U}}_\diamond \geq \varepsilon$. 
    To distinguish the two cases, it suffices to use $O\rbra{d/\varepsilon^2}$ samples of $X$, which gives the upper bound. 

    In comparison, the previous work~\cite{FFGO23} considers the certification task in the incoherent access and ancilla-free setting.
    They obtained upper bounds $O(d/\varepsilon^2)$ and $O(d/\varepsilon^4)$ for trace norm and diamond norm certifications, respectively. They also showed that the upper bound $O(d/\varepsilon^2)$ for trace norm certification is optimal by establishing a matching lower bound $\Omega(d/\varepsilon^2)$, which holds even for ancilla-assisted algorithms.
    We note that in our paper, we consider the incoherent access model but allowing ancilla qubits, i.e., the input to the unknown channel can be an entangled state on the main system and an ancilla system.

\paragraph{Coherent access.}

Inspired by \cite{haah2023query}, our approach for coherent access is obtained by \emph{bootstrapping} the $\varepsilon$ dependence in the incoherent approach. 
First, we note that certifying $\mathcal{E}$ to $\mathcal{U}$ can be reduced to certifying $\mathcal{U}^{-1}\circ\mathcal{E}$ and $\mathcal{I}$. 
Then, instead of measuring immediately after each use of $\mathcal{U}^{-1}\circ \mathcal{E}$, we observe that repeatedly applying $\mathcal{U}^{-1}\circ \mathcal{E}$ can enlarge the promised gap \emph{linearly} when the number of repetitions is small. 
Specifically, as given by \cref{lemma-222209}, the $n$-fold composition of $\mathcal{U}^{-1}\circ \mathcal{E}$, satisfies
\[
\Abs*{\rbra*{\mathcal{U}^{-1}\circ \mathcal{E}}^n-\mathcal{I}}_\diamond > \frac{1}{2}n\Abs*{\rbra*{\mathcal{U}^{-1}\circ \mathcal{E}}-\mathcal{I}}_\diamond,
\]
as long as $n \leq \rbra{2\Abs{\rbra{\mathcal{U}^{-1}\circ \mathcal{E}}-\mathcal{I}}_\diamond}^{-1}$.
Based on this, we can amplify the difference between $\mathcal{U}^{-1}\circ\mathcal{E}$ and $\mathcal{I}$, say $\eta$, to $\Theta(1)$ using only $O(1/\eta)$ queries, and then we can call the incoherent certification algorithm with constant precision.
However, since we do not know the value of $\eta$, we use a step-wise amplification strategy from low precision to high precision with a logarithmic number of total steps. 
Through a careful error analysis, we can conclude that the overall complexity is $O(d/\varepsilon)$.

\paragraph{Source-code access.}
Our approach for source-code access is also obtained by adapting the incoherent approach. 
Recall that in the incoherent approach, a random variable $X$ obtained using one query to $\mathcal{E}$ satisfies (i) $\E\sbra{X} = 1$ when $\mathcal{E} = \mathcal{U}$ while (ii) $\E\sbra{X} \leq 1 - \Theta\rbra{{\varepsilon^2}/{d}}$ when $\Abs{\mathcal{E}-\mathcal{U}}_\diamond \geq \varepsilon$. 
Given source-code access, we view this process as a quantum circuit $V$ from the perspective of amplitudes, which is of the form
\[
V \ket{0} = \sqrt{\E\sbra{X}} \ket{\phi_{0}} + \sqrt{1-\E\sbra{X}} \ket{\phi_{1}},
\]
with $\ket{\phi_0} \perp \ket{\phi_1}$. 
To distinguish the two cases, we focus on the amplitude $\sqrt{1-\E\sbra{X}}$ of $\ket{\phi_1}$, which is $0$ if $\mathcal{E} = \mathcal{U}$ and $\Theta\rbra{\varepsilon/\sqrt{d}}$ if $\Abs{\mathcal{E}-\mathcal{U}}_\diamond \geq \varepsilon$. 
This special case can be determined with query complexity $O\rbra{\sqrt{d}/\varepsilon}$ by quantum amplitude estimation \cite{BHMT02} if given source-code access. 

\subsection{Discussion}

In this paper, we settle the query complexities for quantum channel certification to unitary under three access models with increasing power:
$\Theta(d/\varepsilon^2)$ for incoherent access, $\Theta(d/\varepsilon)$ for coherent access, and $\Theta(\sqrt{d}/\varepsilon)$ for source-code access.
This establishes a strict complexity hierarchy across these access models.
We conclude by listing several directions for future work.
\begin{itemize}
    \item 
    First, an immediate question is whether one can establish a strict hierarchy for other quantum learning and testing problems across these access models.
    \item 
    Second, it would be interesting to identify problems for which the hierarchy collapses; i.e., 
    where a stronger access model does not help in improving the complexity.
    \item
    Third, a broader direction is to investigate whether a larger strict hierarchy can be obtained by considering a richer set of access models in quantum computing.
\end{itemize}

\subsection{Organization}

\cref{sec:incoh} presents the incoherent approach with the analytic connection between entanglement fidelity and the
diamond norm. 
\cref{sec:coh} presents the coherent bootstrapping approach.
\cref{sec:code} presents the approach for source-code access. 

\section{Incoherent access} \label{sec:incoh}
\subsection{The algorithm}
\begin{algorithm}[!htp]
\caption{Incoherent Certification: \texttt{IncohCert}$(\varepsilon,\delta,\mathcal{E},\mathcal{U})$}\label{alg-12231608}
    \begin{algorithmic}[1]
    \Require Precision $\varepsilon$, fail probability $\delta$, queries to the unknown channel $\mathcal{E}$ and classical description of the unitary channel $\mathcal{U}$.
    \Ensure Output either $\mathsf{accept}$ ($\mathcal{E}=\mathcal{U}$) or $\mathsf{reject}$ ($\|\mathcal{E}-\mathcal{U}\|_\diamond\geq \varepsilon$).
    \State $n \gets \ceil{8 d\ln\rbra{1/\delta}/\varepsilon^2}$.
    \For {$i = 1 \textup{ to } n$} 
    \State Perform $(\mathcal{U}^{-1}\circ \mathcal{E}) \otimes \mathcal{I}$ on the maximally entangled state $\ket{\Phi}$. \label{line1}
    \State Perform the POVM $\{M_0=\ketbra{\Phi}{\Phi}, M_1=I-\ketbra{\Phi}{\Phi}\}$ and let $x_i\in\{0,1\}$ be the outcome. \label{line2}
    \EndFor

    \If {all $x_{i}$ are $0$}
        \State \Return $\mathsf{accept}$. 
    \Else
        \State \Return $\mathsf{reject}$. 
    \EndIf
    \end{algorithmic}
\end{algorithm}

\begin{theorem}\label{thm-222140}
Let $\mathcal{E}$ be an unknown $d$-dimensional quantum channel and $\mathcal{U}$ be a known $d$-dimensional unitary channel. Then, \cref{alg-12231608} uses $n=\ceil{8d\ln(1/\delta)/\varepsilon^2}$ queries to $\mathcal{E}$ and distinguishes the cases: (i) $\mathcal{E}=\mathcal{U}$ or (ii) $\|\mathcal{E}-\mathcal{U}\|_\diamond \geq \varepsilon$, with probability at least $1-\delta$.
Moreover, for the case $\mathcal{E}=\mathcal{U}$, the algorithm always outputs ``$\mathcal{E}=\mathcal{U}$'' with probability $1$.
\end{theorem}

\begin{proof}
The algorithm is shown in \cref{alg-12231608}.
If $\mathcal{E}=\mathcal{U}$, then all $x_i$ must be $0$, and \cref{alg-12231608} outputs ``$\mathcal{E}=\mathcal{U}$'' with probability $1$.

Otherwise, we assume $\|\mathcal{E}-\mathcal{U}\|_\diamond \geq \varepsilon$. 
By \cref{lemma-1181906}, this means 
\[\mathrm{F}_{\textup{ent}}(\mathcal{E},\mathcal{U})\leq 1-\frac{1}{8d}\|\mathcal{E}-\mathcal{U}\|_\diamond^2\leq 1-\frac{\varepsilon^2}{8d}.\] 
Thus, for each $i$, the probability of $x_i=0$ is
\[\tr\!\left(\ketbra{\Phi}{\Phi}\rbra{(\mathcal{U}^{-1}\circ\mathcal{E})\otimes \mathcal{I}}(\ketbra{\Phi}{\Phi})\right)= \frac{1}{d^2}\tr(C_{\mathcal{U}}\cdot C_{\mathcal{E}})=\mathrm{F}_{\mathrm{ent}}(\mathcal{E},\mathcal{U})\leq 1-\frac{\varepsilon^2}{8d}.\]
Therefore, for $n \geq 8 d\ln\rbra{1/\delta}/\varepsilon^2$, the probability that there is at least one $x_i = 1$ is at least
\[
1 - \rbra*{1 - \frac{\varepsilon^2}{8d}}^n \geq 1 - \exp\rbra*{-\frac{n\varepsilon^2}{8d}} = 1 - \delta,
\]
where we use the fact that $1 - x \leq \exp\rbra{-x}$ for $x \in \rbra{0, 1}$. 
\end{proof}

\subsection{Technical lemmas}

\begin{lemma}\label{lemma-1181906}
    Let $\mathcal{E}$ be a $d$-dimensional quantum channel and $\mathcal{U}$ be a $d$-dimensional unitary channel. 
    Then, 
    \[
    \sqrt{\mathrm{F}_{\textup{ent}}\rbra{\mathcal{E}, \mathcal{U}}} \leq 1 - \frac{1}{8d} \Abs{\mathcal{E} - \mathcal{U}}_\diamond^2.
    \]
\end{lemma}

\begin{proof}

Let $ \mathcal{F} \coloneqq \mathcal{U}^{-1}\circ \mathcal{E} $. Then, by unitary invariance, we have $ \|\mathcal{E}-\mathcal{U}\|_\diamond = \|\mathcal{F}-\mathcal{I}\|_\diamond $ and $ \mathrm{F}_{\textup{ent}}(\mathcal{E},\mathcal{U}) = \mathrm{F}_{\textup{ent}}(\mathcal{F},\mathcal{I}) $.
Therefore, we only have to show that
\[
\sqrt{\mathrm{F}_{\textup{ent}}\rbra{\mathcal{F}, \mathcal{I}}} \leq 1 - \frac{1}{8d} \Abs{\mathcal{F} - \mathcal{I}}_\diamond^2.
\]
Let $\{A_k\}$ be a set of Kraus operators of $\mathcal{F}$ (no more than $d^2$). 
Let $ V=\sum_k A_k \otimes \ket{k}: \mathbb{C}^d\to \mathbb{C}^d \otimes \mathcal{H}_\mathrm{anc}\cong \mathbb{C}^d\otimes\mathbb{C}^{d^2}$ be a Stinespring dilation isometry for $ \mathcal{F} $, i.e.,
$\mathcal{F}(\rho)=\tr_{\mathrm{anc}}(V\rho V^\dagger)$, where $\cbra{\ket{k}}$ is the computational basis of $\mathcal{H}_\mathrm{anc}$.
Let $W_\eta \colon \ket{\psi} \mapsto \ket{\psi} \otimes \ket{\eta}$ be a Stinespring dilation isometry for $ \mathcal{I} $
for some unit vector $ \ket{\eta}\in\mathcal{H}_\mathrm{anc}$.
By the continuity of Stinespring's representation \cite{KSW08},
\begin{equation} \label{eq:diamond-leq-op}
    \|\mathcal{F}-\mathcal{I}\|_\diamond \le 2 \inf_{\ket{\eta}}\|V-W_\eta\|,
\end{equation}
where $\ket{\eta}$ is taken over all unit vectors in $\mathcal{H}_{\mathrm{anc}}$. 

On the other hand,
\begin{align}
\mathrm{F}_{\textup{ent}}\rbra{\mathcal{F}, \mathcal{I}} &=\sum_{k} \bra{\Phi} (A_k\otimes I) \ketbra{\Phi}{\Phi} (A_k^\dag \otimes I)\ket{\Phi} \nonumber\\
&= \frac{1}{d^2}\sum_k \abs*{\tr\rbra{A_k}}^2, \nonumber 
\end{align}
where $\ket{\Phi}$ denotes the maximally entangled state.
Suppose that $\ket{\eta} = \sum_k c_k \ket{k}$ with $\sum_k \abs{c_k}^2 = 1$. 
By the Cauchy--Schwarz inequality, 
\[
\left|\tr\rbra{W_\eta^\dag V}\right| = \left|\sum_k c_k^* \tr\rbra{A_k} \right|\leq \sqrt{\sum_k \abs*{\tr\rbra{A_k}}^2} = d \sqrt{\mathrm{F}_{\textup{ent}}\rbra{\mathcal{F}, \mathcal{I}}},
\]
which also means that there is a unit vector $\ket{\eta_{\star}}$ such that
\begin{equation} \label{eq:tr-eq-d-sqrt}
    \tr\rbra*{W_{\eta_{\star}}^\dag V} = d \sqrt{\mathrm{F}_{\textup{ent}}\rbra{\mathcal{F}, \mathcal{I}}}.
\end{equation}
Using $\Abs{A} \leq \sqrt{\tr\rbra{A^\dag A}}$, we have
\begin{align}
    \Abs{V - W_{\eta_{\star}}}
    & \leq \sqrt{\tr\rbra*{\rbra*{V - W_{\eta_{\star}}}^\dag \rbra*{V - W_{\eta_{\star}}}}} \nonumber \\
    & = \sqrt{\tr\rbra*{V^\dagger V}+\tr\rbra*{W_{\eta_\star}^\dagger W_{\eta_\star}} - 2\Real\rbra*{\tr\rbra*{W_{\eta_\star}^\dag V}}} \nonumber \\
    & = \sqrt{2d - 2\tr\rbra*{W_{\eta_\star}^\dag V}} \label{eq:step1} \\
    & = \sqrt{2d - 2d \sqrt{\mathrm{F}_{\textup{ent}}\rbra{\mathcal{F}, \mathcal{I}}}}, \label{eq:step2}
\end{align}
where \cref{eq:step1} uses the fact that $\tr\rbra{V^\dag V} = \tr\rbra{W_{\eta_\star}^\dag W_{\eta_\star}} = d$ and $\tr\rbra{W_{\eta_{\star}}^\dag V}$ is a real number, and \cref{eq:step2} is due to \cref{eq:tr-eq-d-sqrt}. 
Combining \cref{eq:diamond-leq-op,eq:step2}, we have
\[
\Abs{\mathcal{F} - \mathcal{I}}_\diamond \leq 2 \sqrt{2d - 2d \sqrt{\mathrm{F}_{\textup{ent}}\rbra{\mathcal{F}, \mathcal{I}}}},
\]
which gives 
\[
\sqrt{\mathrm{F}_{\textup{ent}}\rbra{\mathcal{F}, \mathcal{I}}} \leq 1 - \frac{1}{8d} \Abs{\mathcal{F} - \mathcal{I}}_\diamond^2.
\]
\end{proof}

\subsection{Lower bound}

The upper bound in \cref{thm-222140} matches the lower bound for quantum channel certification to identity in \cite{FFGO23}. 
Here, note that quantum channel certification to identity is a special case of quantum channel certification to unitary. 

\begin{lemma} [{\cite[Theorem 1]{FFGO23}}]
    Let $\mathcal{E}$ be an unknown $d$-dimensional unitary quantum channel and $\varepsilon \in \rbra{0, 1}$. 
    Then, any incoherent algorithm requires $\Omega\rbra{d/\varepsilon^2}$ to $\mathcal{E}$ to distinguish the cases: (i) $\mathcal{E} = \mathcal{I}$ or (ii) $\Abs{\mathcal{E}-\mathcal{I}}_\diamond \geq \varepsilon$, with probability $\geq 2/3$. 
\end{lemma}

\section{Coherent access} \label{sec:coh}
\subsection{The algorithm}
\begin{algorithm}
\caption{Coherent certification: \texttt{CohCert}$(\varepsilon,\delta,\mathcal{E},\mathcal{U})$}\label{alg-222055}
\begin{algorithmic}[1]
\Require Precision $\varepsilon$, fail probability $\delta$, queries to the unknown channel $\mathcal{E}$ and classical description of the unitary channel $\mathcal{U}$.
\Ensure Output either $\mathsf{accept}$ ($\mathcal{E}=\mathcal{U}$) or $\mathsf{reject}$ ($\|\mathcal{E}-\mathcal{U}\|_\diamond\geq \varepsilon$).
\State $T\gets \ceil{\log(1/\varepsilon)+1}$.
\For {$j=0$ to $T$}
\State $p_j\gets 2^j$.
\State $\delta_j\gets \delta \cdot 2^{j-T-1}$.
\State $b_j\gets\mathtt{IncohCert}(1/8,\delta_j,(\mathcal{U}^{-1}\circ\mathcal{E})^{p_j},\mathcal{I})$. \Comment{Call \cref{alg-12231608}}
\If {$b_j=\mathsf{reject}$} 
\State \Return $\mathsf{reject}$. 
\EndIf
\EndFor
\State \Return $\mathsf{accept}$. 
\end{algorithmic}
\end{algorithm}

\begin{theorem}\label{thm-230318}
Let $\mathcal{E}$ be an unknown $d$-dimensional quantum channel and $\mathcal{U}$ be a known $d$-dimensional unitary channel. Then, \cref{alg-222055} uses $n=O(d\log(1/\delta)/\varepsilon)$ queries to $\mathcal{E}$ and distinguishes the cases: (i) $\mathcal{E}=\mathcal{U}$ or (ii) $\|\mathcal{E}-\mathcal{U}\|_\diamond \geq \varepsilon$, with probability at least $1-\delta$.
\end{theorem}

\begin{proof}
The algorithm is shown in \cref{alg-222055}.

\textbf{The case $\bm{\mathcal{E}=\mathcal{U}}$.}
In the $j$-th iteration in \cref{alg-222055}, the probability of $b_j=\mathsf{reject}$ is $0$ due to \cref{thm-222140}. 
Therefore, the probability of \cref{alg-222055} outputting $\mathsf{reject}$ is $0$.

\textbf{The case $\bm{\|\mathcal{E}-\mathcal{U}\|_\diamond\geq \varepsilon}$.}
Let $\eta\coloneqq \|\mathcal{E}-\mathcal{U}\|_\diamond=\|\mathcal{U}^{-1}\circ\mathcal{E}-\mathcal{I}\|_\diamond$.

If $\eta\in(1/2, 2]$, then in the $0$-th iteration, $b_0=\mathsf{accept}$ with probability at most $\delta_0\leq \delta$, due to \cref{thm-222140}. This means \cref{alg-222055} outputs $\mathsf{accept}$ with probability at most $\delta$.

Otherwise, assume $\eta \in (2^{-a},2^{-a+1}]$ for an integer $a\geq 2$, and we know that $a\leq \log(1/\varepsilon)+1$.
Note that $T\geq a$ and we consider the $(a-2)$-th iteration.
We can see that $p_{a-2}=2^{a-2}\leq \frac{1}{2\eta}$. Using \cref{lemma-222209}, we have
\[\left\|(\mathcal{U}^{-1}\circ\mathcal{E})^{p_{a-2}}-\mathcal{I}\right\|_\diamond> \frac{1}{2}p_{a-2}\eta=2^{a-3}\eta> \frac{1}{8}.\]
Then, by \cref{thm-222140}, the probability that $b_{a-2}=\mathsf{accept}$ is at most $\delta_{a-2}\leq \delta$.
This means \cref{alg-222055} outputs $\mathsf{accept}$ with probability at most $\delta$.

\textbf{Complexity analysis.} The complexity of \cref{alg-222055} is thus
\begin{align*}
    \sum_{j=0}^T p_j \cdot O\rbra{d \log\rbra{1/\delta_j}}
    & = \sum_{j=0}^T 2^j \cdot O\rbra{d \rbra*{\log\rbra{1/\delta}+T+1-j}} \\
    & = O\rbra{2^T d \log\rbra{1/\delta}} + O\rbra{d\cdot 2^{T+2}} \\
    & = O\rbra{d\log\rbra{1/\delta}/\varepsilon}.
\end{align*}

\end{proof}

\subsection{Technical lemmas}
\begin{lemma}\label{lemma-222209}
    If $\Abs{\mathcal{E} -\mathcal{I}}_\diamond= \varepsilon$, then for any integer $1 \leq n\leq \frac{1}{2\varepsilon}$, we have 
    \[\|\mathcal{E}^n-\mathcal{I}\|_\diamond > \frac{1}{2}n\varepsilon. \]
\end{lemma}
\begin{proof}
Let $\mathcal{E}=\mathcal{I}+\Delta$, where $\|\Delta\|_\diamond = \varepsilon$. Then we have
\[\mathcal{E}^n-\mathcal{I}=(\mathcal{I}+\Delta)^n-\mathcal{I}=n\Delta + \sum_{i=2}^n\binom{n}{i}\Delta^i.\]
Therefore, 
\begin{align}
\|\mathcal{E}^n -\mathcal{I}\|_\diamond &\geq n\|\Delta\|_\diamond-\sum_{i=2}^n \binom{n}{i}\|\Delta^i\|_\diamond \nonumber\\
&\geq n\varepsilon -\sum_{i=2}^n \binom{n}{i} \varepsilon^i \label{eq-1182021}\\
&=n\varepsilon - \left((1+\varepsilon)^n-1 -n\varepsilon\right) \nonumber\\
&\geq n\varepsilon - \left(e^{n\varepsilon}-1-n\varepsilon\right)\nonumber \\ 
&\geq n\varepsilon - e^{n\varepsilon}\frac{(n\varepsilon)^2}{2}\label{eq-tanfu}\\
&\geq n\varepsilon - \frac{e^{1/2}}{4} n\varepsilon \label{eq-1182030} \\
&> \frac{1}{2}n\varepsilon \nonumber
\end{align}
where \cref{eq-1182021} uses the sub-multiplicativity $\|\Delta^k\|_\diamond\leq \|\Delta\|_\diamond^k$ (see \cite[Proposition 3.48(1)]{Wat18}), \cref{eq-tanfu} uses the Taylor expansion 
\[e^{n\varepsilon}-1-n\varepsilon=\sum_{i=2}^\infty \frac{(n\varepsilon)^i}{i!}\leq \frac{(n\varepsilon)^2}{2} \sum_{i=0}^\infty \frac{(n\varepsilon)^i}{i!}=\frac{(n\varepsilon)^2}{2} e^{n\varepsilon},\]
and \cref{eq-1182030} uses $n\varepsilon\leq 1/2$.
\end{proof}

\subsection{Lower bound}

In the following we present a query lower bound of $\Omega(d/\varepsilon)$ for quantum channel certification to identity, a special case of certification to unitary.
The proof follows from a simple reduction to Regev-Schiff's $p$-faulty Grover problem~\cite{RS08}.

\begin{lemma}
    \label{lmm:cert-i-lb}
    Let $\mathcal{E}$ be an unknown $d$-dimensional quantum channel. Suppose $\varepsilon\in (0,1)$. Then it requires $\Omega(d/\varepsilon)$ queries to $\mathcal{E}$ to distinguish the cases: (i) $\mathcal{E}=\mathcal{I}$ or (ii) $\|\mathcal{E}-\mathcal{I}\|_\diamond \geq \varepsilon$, with probability $\geq 2/3$.
\end{lemma}

\begin{proof}
Let us consider the following hard instance from~\cite{RS08}.
For each $k\in\{0,1,\dots,d-1\}$, define the $p$-faulty Grover oracle and its corresponding quantum channel:
\begin{align*}
    O_k &= I - 2\ketbra{k}{k},\\
    \mathcal E_k(\rho)&= p\rho + (1-p)O_k \rho O_k^\dagger,
\end{align*}
where $p=1-\frac{\varepsilon}{2}> \frac{1}{2}$.

Consider the task of distinguishing between the following two cases:
\begin{enumerate}
    \item 
    $\calE =\calI$. 

    \item
    $\calE=\calE_k$ for some $k\in \braces*{0,1,\ldots, d-1}$.
    In this case, pick any $j\neq k$.
    Define
    \[
    \ket{\psi}=\frac{\ket{k}+\ket{j}}{\sqrt2},
    \qquad
    \ket{\psi_\perp}=\frac{-\ket{k}+\ket{j}}{\sqrt2}.
    \]
    It is easy to see that $O_k\ket{\psi}=\ket{\psi_\perp}$ and $\braket{\psi}{\psi_\perp}=0$, which mean
    \[
    \mathcal E_k(\ket{\psi}\!\bra{\psi})
    \;=\;
    p\,\ket{\psi}\!\bra{\psi}+(1-p)\,\ket{\psi_\perp}\!\bra{\psi_\perp}.
    \]
    Then, we can calculate that
    \[
    \norm*{\calE_k - \calI}_{\diamond}\geq 
    \norm*{\calE_k(\ket{\psi}\!\bra{\psi})-\ket{\psi}\!\bra{\psi}}_1
    = 2(1-p) = \varepsilon.
    \]
\end{enumerate}
Any quantum algorithm for quantum channel certification to identity (see \Cref{lmm:cert-i-lb}) can also distinguish the above two cases.
Since the $p$-faulty Grover problem has query lower bound $\Omega\!\left(\frac{pd}{1-p}\right)$~\cite{RS08},
it follows that $\Omega(d/\varepsilon)$ queries are required for quantum channel certification by $p=1-\frac{\varepsilon}{2}$.
\end{proof}

\section{Source-code access} \label{sec:code}

\subsection{The algorithm}

\begin{theorem} \label{thm-1526}
    Let $W$ be the source code (quantum unitary circuit) that implements an unknown $d$-dimensional quantum channel $\mathcal{E}$. 
    Let $\mathcal{U}$ be a known $d$-dimensional unitary channel. 
    Then, \cref{alg-1457} uses $n = O\rbra{\sqrt{d}\log\rbra{1/\delta}/\varepsilon}$ queries to controlled-$W$ and controlled-$W^\dag$, and distinguishes the cases: (i) $\mathcal{E}=\mathcal{U}$ or (ii) $\|\mathcal{E}-\mathcal{U}\|_\diamond \geq \varepsilon$, with probability at least $1-\delta$.
\end{theorem}

To prove \cref{thm-1526}, we need the quantum amplitude estimation \cite{BHMT02}.  
Here, for convenience, we use the version adapted from \cite{Wan24}. 

\begin{theorem}[{\cite[Theorem III.4]{Wan24}}] \label{thm:ae}
    Let $U$ be a unitary operator such that 
    \[
    U \ket{0}_{\mathsf{A}} \ket{0}_{\mathsf{B}} = \sqrt{1-p} \ket{0}_{\mathsf{A}} \ket{\phi_0}_{\mathsf{B}} + \sqrt{p} \ket{\perp}_{\mathsf{A}} \ket{\phi_1}_{\mathsf{B}},
    \]
    where $p \in \sbra{0, 1}$, $\ket{0}_{\mathsf{A}} \perp \ket{\perp}_{\mathsf{A}}$, and $\ket{\phi_0}$ and $\ket{\phi_1}$ are normalized pure states. 
    Then, there is a quantum query algorithm $\textup{\texttt{SqrtAmplEst}}\rbra{\varepsilon, \delta, U}$ that estimates $\sqrt{p}$ to within additive error $\varepsilon$ with success probability $\geq 1 - \delta$ using $O\rbra{\log\rbra{1/\delta}/\varepsilon}$ queries to controlled-$U$ and controlled-$U^\dag$. 
\end{theorem}

\begin{algorithm}[t]
\caption{Source-Code Certification: \texttt{SourceCodeCert}$(\varepsilon,\delta,\mathcal{E},\mathcal{U})$}\label{alg-1457}
    \begin{algorithmic}[1]
    \Require Precision $\varepsilon$, fail probability $\delta$, query access to the source code $W$ of the unknown channel $\mathcal{E}$ such that $\tr_{\mathsf{B}}\rbra{W \rbra{\rho_{\mathsf{A}} \otimes \ketbra{0}{0}_{\mathsf{B}}} W^\dag} = \mathcal{E}\rbra{\rho}$ for any state $\rho$, and classical description of the unitary channel $\mathcal{U}$ (with $U$ denoting its corresponding unitary operator).
    \Ensure Output either $\mathsf{accept}$ ($\mathcal{E}=\mathcal{U}$) or $\mathsf{reject}$ ($\|\mathcal{E}-\mathcal{U}\|_\diamond\geq \varepsilon$).

    \State Let $U_{\Phi}$ be the state-preparation circuit of the maximally entangled state $\ket{\Phi}_{\mathsf{AC}}$. 
    
    \State Let $V = \rbra{U_{\Phi}^\dag \otimes I_{\mathsf{B}}} \cdot \rbra{ \rbra{ \rbra{U^\dag_{\mathsf{A}} \otimes I_{\mathsf{B}}} \cdot W_{\mathsf{AB}}} \otimes I_{\mathsf{C}} } \cdot \rbra{U_{\Phi} \otimes I_{\mathsf{B}}}$.

    \State $p \gets \texttt{SqrtAmplEst}\rbra{\frac{\varepsilon}{16\sqrt{d}}, \delta, V}$. \Comment{Use \cref{thm:ae}}

    \If {$p < \frac{\varepsilon}{4\sqrt{d}}$} 
    \State \Return $\mathsf{accept}$.
    \Else
    \State \Return $\mathsf{reject}$.
    \EndIf
    
    \end{algorithmic}
\end{algorithm}

\begin{proof}[Proof of \cref{thm-1526}]
    The construction of $V$ follows Lines \ref{line1} and \ref{line2} in \cref{alg-12231608}, and we have 
    \[
    \Abs*{ \bra{0}_{\mathsf{AC}} V_{\mathsf{ABC}} \ket{0}_{\mathsf{ABC}} }^2 = \tr\!\left(\ketbra{\Phi}{\Phi}\rbra{(\mathcal{U}^{-1}\circ\mathcal{E})\otimes \mathcal{I}}(\ketbra{\Phi}{\Phi})\right) = \mathrm{F}_{\textup{ent}}(\mathcal{E},\mathcal{U}).
    \]
    Therefore, by \cref{thm:ae}, $p$ is an estimate of $\sqrt{1 - \mathrm{F}_{\textup{ent}}(\mathcal{E},\mathcal{U})}$ such that
    \[
    \Pr\sbra*{ \abs*{p - \sqrt{1 - \mathrm{F}_{\textup{ent}}(\mathcal{E},\mathcal{U})}} \leq \frac{\varepsilon}{16\sqrt{d}} } \geq 1 - \delta.
    \]
    
    Then, we follow the argument in the proof of \cref{thm-222140}.
    If $\mathcal{E}=\mathcal{U}$, then $\mathrm{F}_{\textup{ent}}(\mathcal{E},\mathcal{U}) = 1$, and thus $p < \frac{\varepsilon}{4\sqrt{d}}$ with probability $\geq 1 - \delta$, in which case \cref{alg-1457} outputs ``$\mathcal{E}=\mathcal{U}$'' with probability $\geq 1 - \delta$. 
    Otherwise, if $\|\mathcal{E}-\mathcal{U}\|_\diamond \geq \varepsilon$, then by \cref{lemma-1181906},  
\[\sqrt{1-\mathrm{F}_{\textup{ent}}(\mathcal{E},\mathcal{U})} \geq \sqrt{\frac{1}{8d}\|\mathcal{E}-\mathcal{U}\|_\diamond^2} \geq \frac{\varepsilon}{\sqrt{8d}},\]
and thus $p > \frac{\varepsilon}{4\sqrt{d}}$ with probability $\geq 1 - \delta$, in which case outputs $\mathsf{reject}$ with probability $\geq 1 - \delta$. 

    To complete the proof, the query complexity is $O\rbra{\sqrt{d}\log\rbra{1/\delta}/{\varepsilon}}$ due to the use of \cref{thm:ae}. 
\end{proof}

\subsection{Lower bound}

The upper bound in \cref{thm-1526} matches the lower bound for quantum channel certification to identity in \cite{JO26}. 
Here, note that quantum channel certification to identity is a special case of quantum channel certification to unitary. 

\begin{lemma} [{\cite[Theorem 2]{JO26}}]
    Let $\mathcal{E}$ be an unknown $d$-dimensional unitary quantum channel and $\varepsilon \in \rbra{0, 1/2}$. 
    Then, it requires $\Omega\rbra{\sqrt{d}/\varepsilon}$ to $\mathcal{E}$ and $\mathcal{E}^{-1}$ to distinguish the cases: (i) $\mathcal{E} = \mathcal{I}$ or (ii) $\Abs{\mathcal{E}-\mathcal{I}}_\diamond \geq \varepsilon$, with probability $\geq 2/3$. 
\end{lemma}

\bibliographystyle{alphaurl}
\bibliography{main}

\end{document}